\documentclass[twoside,leqno,twocolumn]{article}
\usepackage{ltexpprt}
\usepackage{url}  
\usepackage{graphicx}  
\usepackage{mathrsfs}
\usepackage{mdframed}
\usepackage{amsmath}
\usepackage{xcolor}
\usepackage{float}
\usepackage{subcaption}
\usepackage{algpseudocode}
\usepackage{algorithm}
\usepackage{ntheorem}
\usepackage[colorlinks = true,
            linkcolor = black,
            urlcolor  = blue,
            citecolor = black,
            anchorcolor = black]{hyperref}
\usepackage{enumitem}
\newtheorem{prop}{Proposition}

\newtheorem{defn}{Definition}

\begin{document}

\title{\Large Edge Replacement Grammars : A Formal Language Approach for Generating Graphs}
\author{Revanth Reddy$^1$\thanks{Both authors contributed equally.}  \qquad \qquad Sarath Chandar$^2$\footnotemark[1] \qquad \qquad Balaraman Ravindran$^{1,3}$ \\
$^1$Department of Computer Science and Engineering, Indian Institute of Technology Madras \\
$^2$Mila, Universit\'e de Montr\'eal \\
$^3$Robert Bosch Centre for Data Science and AI, Indian Institute of Technology Madras\\
\texttt{g.revanthreddy111@gmail.com, apsarathchandar@gmail.com, ravi@cse.iitm.ac.in}}

\date{}
\maketitle


\fancyfoot[R]{\scriptsize{Copyright \textcopyright\ 2019 by SIAM\\
Unauthorized reproduction of this article is prohibited}}



\setlength{\belowcaptionskip}{-10pt}
\renewcommand{\figurename}{Fig.}

\begin{abstract} \small\baselineskip=9pt 

Graphs are increasingly becoming ubiquitous as models for structured data. A generative model that closely mimics the structural properties of a given set of graphs has utility in a variety of domains. Much of the existing work require that a large number of parameters, in fact exponential in size of the graphs, be estimated from the data. We take a slightly different approach to this problem, leveraging the extensive prior work in the formal graph grammar literature. In this paper, we propose a graph generation model based on Probabilistic Edge Replacement Grammars (PERGs). We propose a variant of PERG called Restricted PERG (RPERG), which is analogous to PCFGs in string grammar literature. With this restriction, we are able to derive a learning algorithm for estimating the parameters of the grammar from graph data. We empirically demonstrate on real life datasets that RPERGs outperform existing methods for graph generation. We improve on the performance of the state-of-the-art Hyperedge Replacement Grammar based graph generative model. Despite being a context free grammar, the proposed model is able to capture many of the structural properties of real networks, such as degree distributions, power law and spectral characteristics.\\
\textbf{Keywords: }Graph Generative Models, Graph Mining, Graph Grammars\end{abstract}

\section{Introduction}

Graphs are used to represent various structured data. A variety of networks ranging from social networks to biological networks can be represented as graphs with nodes representing entities and edges representing the relationship between them. Because of the widespread use of graphs as a representation language, many of the usual machine learning tasks are now being specialized for graphs.

One such machine learning task is to estimate the parameters of the generative model of a graph. A good generative model should be able to capture the structural properties of the graph, like degree distribution, community structure, smaller diameter, eigen distributions and so on.
The advantages of having a good generative model for a class of graphs are several-fold:
\begin{itemize}[noitemsep]
    \item We can use the generative model to generate realistic graphs and run simulation studies on it, instead of running experiments on the real network, which might not be feasible always.
    
    
    \item If we are able to fit the model more accurately, we can use the model to compress the graph data, by just saving the model instead of the entire graph data.
    
    \item We can do graph classification if we can learn a generative model for a class of graphs, by determining the notion of likelihood of the test graph as per the given model.
    
    \item The model can also be used to anonymize the graph data, by generating graphs similar to the original graphs and keeping the original graphs confidential. This will be more helpful for medical data.
    
\end{itemize}

These advantages make the problem of designing generative models for graphs an important research problem in network sciences and various graph generation models have been proposed in the past. The earliest generative model in a probabilistic setting was the E-R random graph model \cite{1} . However, the model fails to match several network properties. Specifically, this model does not simulate heavy tailed degree distributions. To overcome this, several other models were proposed. Most of these models belong to the family of preferential attachment models \cite{2,3} which employ the ``rich get richer" phenomenon, which leads to power law distributions. There are several variations of ``rich get richer" models like the ``copying model" \cite{6}, the ``winner does not take all" model \cite{7}, the ``forest fire" model \cite{leskovec2007graph} and so on. There is also a different class of models that simulate the ``small world network" \cite{9}. For a detailed survey of the existing statistical network models, refer \cite{10}. 

Most of these models match one or few of the properties of the natural graph. There has been significant interest to come up with a single model that can simulate most of the graph properties. Kronecker graph generators \cite{11} is an example. However they have few limitations. For example, the number of nodes is predetermined. A recursive realistic graph generator using random typing is proposed in \cite{12}. Even though there are many such models, designing a model which has a fast and scalable learning procedure, while also capturing all the structural properties of the network is still a challenging problem.

In this work, we propose a graph generation model based on graph grammars. Unlike other graph generation models, we view graph generation process as a formal language derivation process. We assume that there is an underlying grammar which is generating this graph and the graph evolves according to the grammar rules. So, the problem of graph generation is now reduced to that of inducing the grammar which generated this data, allowing us to leverage extensive prior work in the formal graph grammar literature. Here, the graph generation process is viewed as a derivation from a single edge using a probabilistic edge replacement grammar. The likelihood of a graph belonging to a particular family is the probability of the derivation under the appropriate grammar. The idea of using graph grammars for graph generation was also explored in \cite{aguinaga2016growing} where authors propose a hyperedge replacement grammar (HRG) based generative model. We compare our approach with their approach.

Edge Replacement Grammars are graph grammar formalisms where the rules replace an edge in a graph with another graph. We propose a variant of Probabilistic Edge Replacement Grammar called Restricted Probabilistic Edge Replacement Grammar (RPERG). RPERG is analogous to PCFG in string grammar literature. This will become evident once we define RPERGs formally. We tested the capabilities of the model by fitting it onto several real world datasets. Experimental results demonstrate that the model is able to capture most of the statistical and structural properties of the graph better than existing graph generators. The major advantages of this model over the existing models are as follows:
\begin{itemize}[noitemsep]
    \item The model makes no assumptions on the underlying graph family.
    
    \item The model assumes no specific parametric form. The parameters of this model are the grammar rules and the number of rules is determined by the
complexity of the data.

    \item The model parameters are more easily interpretable. They are nothing but
the statistically significant subgraph patterns that repeat itself in the graph. They can also be considered as the motifs in the graph.
\end{itemize}
Contributions of this paper are several-fold:
\begin{itemize}[noitemsep]
    \item  We define a family of graphs called ``non-squeezable graphs" and provide a complete characterization of the family.
    \item  We propose a PCFG equivalent grammar in graph grammar literature based on non-squeezable graphs, which we call as Restricted Probabilistic Edge Replacement Grammar (RPERG).
    \item We provide a maximum-likelihood learning methodology to learn the grammar from the given data.
    \item The proposed model captures the structural properties of the graph better than existing state-of-the-art graph generators.
\end{itemize}
The rest of the paper is organized as follows. In Section 2, we define the basic terminology related to Edge Replacement Grammars and briefly describe the existing Hyperedge Replacement Grammar (HRG)\cite{aguinaga2016growing} based approach. Section 3 introduces the family of non-squeezable graphs and gives an algorithm for learning RPERGs from a set of graphs. In Section 4, we provide results for performance of the proposed models on various datasets. Section 5 concludes the paper and gives directions for future work.

\section{Background}
\subsection{Edge Replacement Grammars}

We define Edge Replacemet Grammars (ERGs) along the lines of Hyperedge Replacement Grammars (HRGs) by \cite{13}. For the sake of simplicity, we state our definitions in terms of edge labeled undirected graphs. The concepts can be easily extended to accommodate node labels as well as directed edges.

\begin{defn}
An edge replacement grammar (ERG) is a tuple $\mathcal{G}$= $\langle N,T,P,S \rangle$ where
   \begin{itemize}[noitemsep]
       \item N and T are finite disjoint sets of non-terminal and terminal edge labels.
       \item S $\in$ N is the start edge label.
       \item P is a finite set of productions of the form A $\rightarrow$ R, where $A\in$ N and R is a graph with edge labels drawn from N $\cup$ T.
   \end{itemize}
\end{defn}


We say that a graph $X'$ is derived from a graph $X$ in ERG $\mathcal{G}$, if we can obtain $X'$ by applying a series of production rules starting from $X$. We denote this by \(X \Longrightarrow_\mathcal{G}^* X' \). 
Figure \ref{fig:fig1a} gives an example ERG and Figure \ref{fig:fig1b} gives a sample derivation using the grammar. Another important thing to note is that the paper makes the assumption that T = $\{\epsilon\}$ i.e all the edge labels are non-terminal edge labels. 

\begin{figure}[h!]
	\centering
	\captionsetup{skip=12pt}
	\begin{subfigure}[b]{.22\textwidth}
		\centering
	\captionsetup{skip=-3pt}	\includegraphics[scale=0.27]{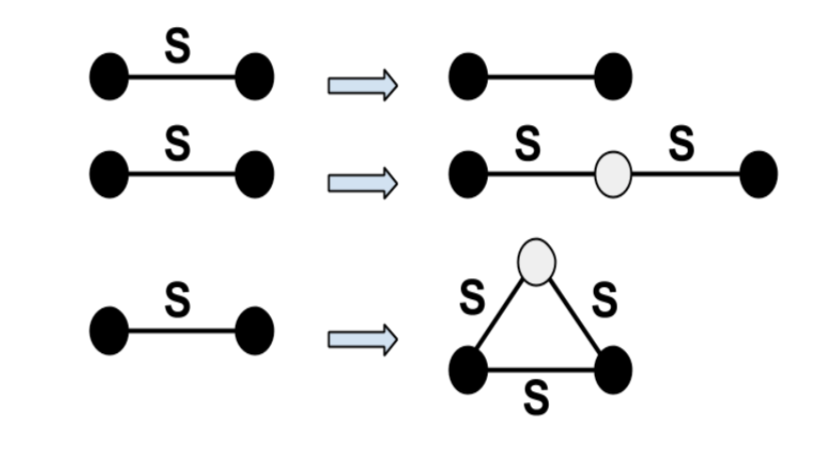}
	    \caption{}\label{fig:fig1a}
	\end{subfigure}
	\begin{subfigure}[b]{.22\textwidth}
		\centering
	\captionsetup{skip=-3pt}	\includegraphics[scale=0.27]{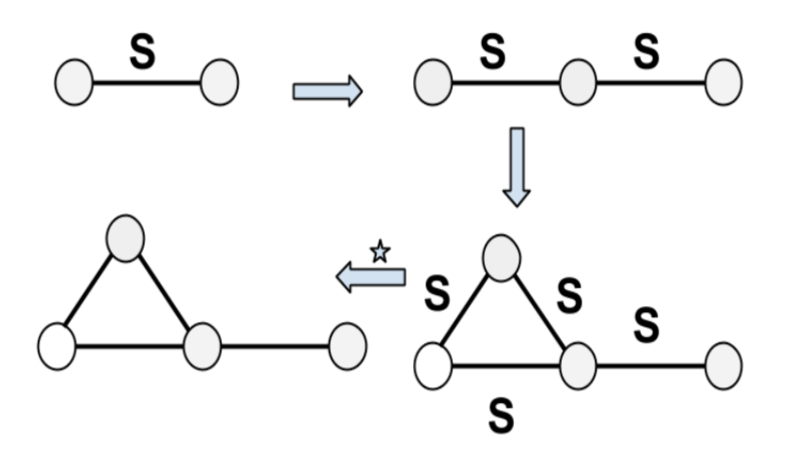}
	    \caption{}\label{fig:fig1b}
	\end{subfigure}%
	\caption{\small (a) Sample ERG (b) Sample derivation using Grammar in (a) }
	
\end{figure}
\begin{defn}
A Probabilistic Edge Replacement Grammar (PERG) consists of 
    \begin{itemize}[noitemsep]
        \item An edge replacement grammar $\mathcal{G}$ = $\langle N,T,P,S \rangle$ 
        \item A parameter p(A$\rightarrow$R) for each rule $A\rightarrow R \in P$, which is the conditional probability of choosing this rule given that the non-terminal being expanded is A. For any $X \in N$, $\sum_{A\rightarrow R:A=X} p(A \rightarrow R) = 1$
    \end{itemize}
\end{defn}

Let $G_\mathcal{G}$ be the set of all graphs that can be generated from the grammar $\mathcal{G}$. For any graph $g \in G_\mathcal{G}$ generated by applying the rules $A_1 \rightarrow R_1$, $A_2 \rightarrow R_2$,...,$A_n \rightarrow R_n$, the probability of $g$ under PERG is given by 
\setlength{\abovedisplayskip}{3pt}
\setlength{\belowdisplayskip}{3pt}
\[p(g) = \prod_{i=1}^n p(A_i \rightarrow R_i)\]

If we assign probabilities 0.2, 0.4, 0.4 to the three rules in Figure-1a respectively, then the probability of the graph generated in Figure-1b is given by $0.4*0.4*(0.2)^4$. The sum of probabilities of all $g \in G_\mathcal{G}$ will be 1. Here, probability of $g$ under $\mathcal{G}$ is the probability of generating the graph $g$  by sampling rules from the grammar $\mathcal{G}$.

\subsection{HRG based approach}
\label{sec:hrg}
HRG based graph generative model \cite{aguinaga2016growing} has been shown to outperform existing Chung-Lu \cite{chung2002connected} and Kronecker \cite{leskovec2010kronecker} models. In this section, we give a brief overview of the HRG based approach. First, we introduce clique trees and then define hyperedge replacement grammars. The content in this section is based on \cite{aguinaga2016growing}.
	
    All graphs can be decomposed into a clique tree. A network's clique tree encodes robust and precise information about the network. Here, we just give a brief definition of clique trees. For more information, we refer the reader to Chapters 9,10 of \cite{Koller}. 
	
    \begin{defn}
   A clique tree of a graph H = (V,E) is a tree T, each of whose nodes $\eta$ is labelled with a $V_{\eta} \subseteq V$ and $E_{\eta} \subseteq E$, such that the following properties hold:
   \begin{itemize}[noitemsep]
       \item Vertex Cover: For each  $v \in V$, there is a vertex $\eta \in T$ such that $v \in V_{\eta}$.
       \item Edge Cover: For each hyperedge $e_i = \{v_1,...,v_k\} \in E$, there is exactly one node $\eta \in T$ such that $e \in E_{\eta}$. Moreover, $v_1,...,v_k \in V_{\eta} $
   \end{itemize}
   \end{defn}
   
   
    
   A hyperedge is an edge which can connect any number of vertices. If a hyperedge edge $e$ connects vertices $v_1,v_2,...,v_i$, then it is denoted as: \(e = \{v_1,v_2,...,v_i\} \) . Here, \(|e| = i\). A hypergraph is a graph $H = (V,E)$ in which each edge is a hyperedge. 
    
\begin{defn}
   A hyperedge replacement grammar is a tuple $\mathcal{G}$ = $\langle N,T,S,P \rangle$, where
   \begin{itemize}[noitemsep]
       \item N is a finite set of non-terminal symbols. Each nonterminal A has a non negative integer rank, which we write $|e|$.
       \item T is a finite set of terminal symbols. 
       \item $S \in N$ is a distinguished starting nonterminal, and $|S| = 0$
       \item P is a finite set of production rules $A \rightarrow R$, where
        \begin{itemize}[noitemsep]
            \item A is a nonterminal symbol.
            \item R is a hypergraph whose edges are labelled by symbols from $T \cup N$. If an edge e is labelled by a non-terminal B, we must have $|e| = |B|$.
            \item Exactly $|A|$ vertices of R are designated \textit{external} vertices. The other vertices in R are called \textit{internal} vertices.
        \end{itemize}
   \end{itemize}
   
   \end{defn}
      
   The first step in learning an HRG from a graph is to compute a clique tree from the original graph. Finding the minimal-width clique tree is NP-complete \cite{arnborg1987complexity}. \cite{aguinaga2016growing} uses a Maximum Cardinality Search (MCS) heuristic introduced by \cite{tarjan1984simple} to compute a clique tree with a reasonably-low, but not necessarily minimal, width. Then, this clique tree induces an HRG in a natural way as shown below. The approach differs based on the type of node of the clique tree that is being processed. We refer the reader to \cite{aguinaga2016growing} for a more detailed discussion and visualization of the HRG learning process. 
   
   \begin{itemize}
   \item \textbf{Interior Node:} Let $\eta$ be an interior node of the clique tree $T$, let $\eta'$ be its parent, and let $\eta_1, \eta_2,...,\eta_m$ be its children. Node $\eta$ corresponds to an HRG production rule $A \rightarrow R$ as follows. First, $|A|$ = $|V_{\eta'} \cap V_{\eta}|$. Then, $R$ is formed by:
   \begin{itemize}[noitemsep]
       \item Adding an isomorphic copy of the vertices in $V_{\eta}$ and the edges in $E_{\eta}$.
       \item Marking the (copies of) vertices in $V_{\eta'} \cap V_{\eta}$ as external vertices.
       \item Adding, for each $\eta_i$, a nonterminal hyperedge connecting the (copies of) vertices in $V_{\eta} \cap V_{\eta_i}$
   \end{itemize}

   \item \textbf{Root Node:} The RHS is computed similar to the interior node case except that it has no external vertices. The start non-terminal $S$ is the LHS and it has rank 0.
   
   \item \textbf{Leaf Node:} The LHS and RHS are calculated in the same way as the interior node case except that no new non-terminal hyperedges are added to the RHS, as there are no children.
   
   \end{itemize}
 
\section{Approach}

\subsection{Non-squeezable graphs}

Learning PERG from the graph data is hard, since the RHS of rules can be any subgraph. So we define a restricted version of PERGs, which we call Restricted PERGs (RPERGs). Before defining RPERG, we introduce a new operation in connected graphs, called \textit{squeezing}.
\begin{defn}
Let $u$, $v$ be a pair of vertices in the graph $G$. Let $g_1, g_2, ..., g_t$ be the connected components obtained by removing $u,v$ from $G$. A squeezing operation with respect to $u, v$ is an operation where one of the components $g_i$ is replaced by an edge between $u, v$.
\end{defn}

Here, $t$ is the number of connected components obtained after removing $u$, $v$ from G. When $t = 1$, the entire graph will be squeezed into a single edge. Squeezing can be viewed as the reverse operation of edge expansion. 
The following is a special case for the squeezing operation. If $t \geq 3$ and $g_1,...,g_t$ are isolated vertices, then the squeeze operation with respect to $u, v$ replaces the entire graph with the edge $u, v$. Figure \ref{fig:fig2} gives some examples for squeezing.

A squeezing operation in which the entire graph is squeezed into a single edge is called a \textit{trivial squeeze}. Now we will define a class of graphs called non-squeezable graphs.

\begin{figure}[h!]
	\centering
	\captionsetup{skip=14pt}
	\begin{subfigure}[b]{.15\textwidth}
		\centering
	\captionsetup{skip=0pt}	\includegraphics[scale=0.25]{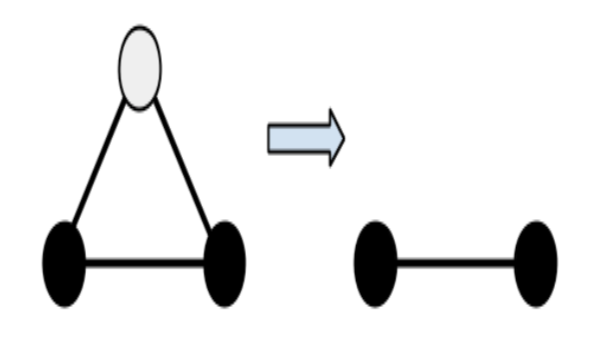}
	    \caption{}\label{fig:fig2a}
	\end{subfigure}
	\begin{subfigure}[b]{.15\textwidth}
		\centering
	\captionsetup{skip=0pt}	\includegraphics[scale=0.25]{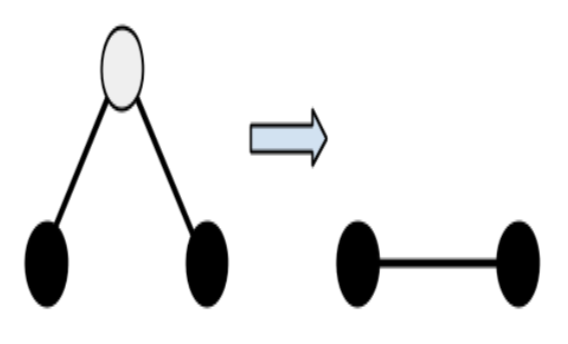}
	    \caption{}\label{fig:fig2b}
	\end{subfigure}%
	\begin{subfigure}[b]{.15\textwidth}
		\centering
	\captionsetup{skip=0pt}	\includegraphics[scale=0.25]{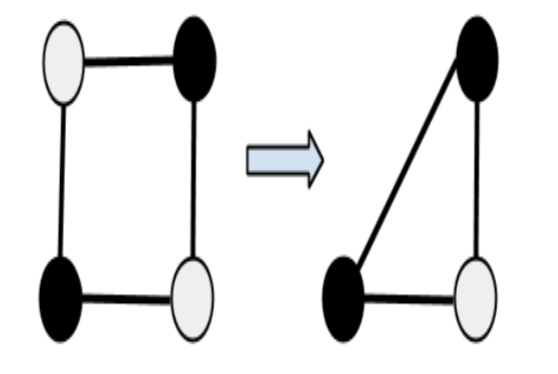}
	    \caption{}\label{fig:fig2c}
	\end{subfigure}%
	\caption{\small (a),(b) are trivial squeeze while (c) is a non-trivial squeeze. In (b), a new edge has been introduced due to squeezing. The darkened nodes correspond to $u, v$.}
    \label{fig:fig2}
\end{figure}

\begin{defn}
A non-squeezable graph is a graph in which the only squeeze operation that is possible is the trivial squeeze.
\end{defn}

A graph is squeezable if there are non-trivial squeezes possible. Figure \ref{fig:fig3} gives examples for some non-squeezable graphs and squeezable graphs. Triangle and star graphs are considered to be degenerate cases for non-squeezable graphs. We will now try to characterize the class of graphs that are non-squeezable.

\begin{figure}[h!]
	\centering
	\captionsetup{skip=14pt}
	\begin{subfigure}[b]{.075\textwidth}
		\centering
	\captionsetup{skip=0pt}	\includegraphics[scale=0.17]{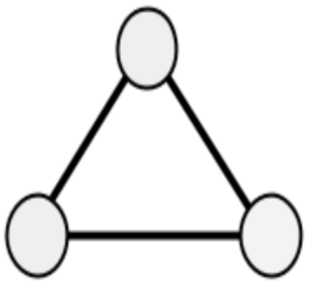}
	    \caption{}\label{fig:fig2a}
	\end{subfigure}
	\begin{subfigure}[b]{.075\textwidth}
		\centering
	\captionsetup{skip=0pt}	\includegraphics[scale=0.17]{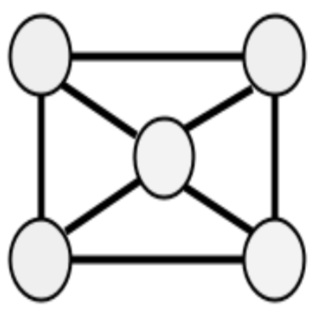}
	    \caption{}\label{fig:fig2b}
	\end{subfigure}%
	\begin{subfigure}[b]{.075\textwidth}
		\centering
	\captionsetup{skip=0pt}	\includegraphics[scale=0.17]{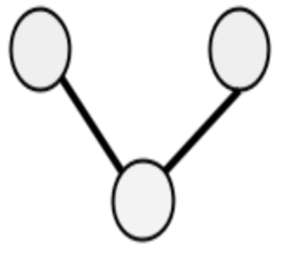}
	    \caption{}\label{fig:fig2c}
	\end{subfigure}%
	\begin{subfigure}[b]{.075\textwidth}
		\centering
	\captionsetup{skip=0pt}	\includegraphics[scale=0.17]{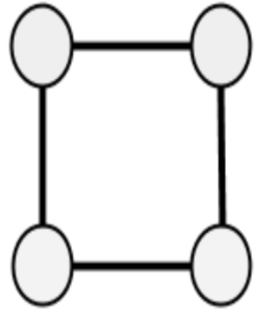}
	    \caption{}\label{fig:fig2a}
	\end{subfigure}
	\begin{subfigure}[b]{.075\textwidth}
		\centering
	\captionsetup{skip=0pt}	\includegraphics[scale=0.17]{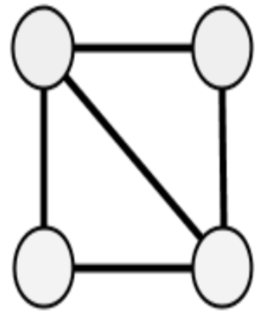}
	    \caption{}\label{fig:fig2b}
	\end{subfigure}%
	\begin{subfigure}[b]{.075\textwidth}
		\centering
	\captionsetup{skip=0pt}	\includegraphics[scale=0.17]{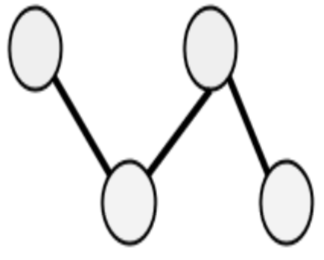}
	    \caption{}\label{fig:fig2c}
	\end{subfigure}%
    \caption{\small a,b,c are non-squeezable while d,e,f are squeezable.}
    \label{fig:fig3}
\end{figure}

\begin{prop}
All $k$-vertex connected graphs for $k \geq 3$ are non-squeezable.
\end{prop}

\begin{proof}
The proof is based on the definition of \textit{squeezing} operation. For any $k$-connected graph with $k\geq3$, we need atleast 3 vertices to disconnect the graph into two components. \textit{Squeeze} operation essentially finds a partition of the graph into two parts and squeezes one of them into an edge. This is not possible when $k\geq3$, since you cannot find a pair of vertices that partitions the graph into two parts. Note that the reverse of this proposition is not true. Figure 3-c is a counter-example which is 1-connected and non-squeezable.
\end{proof}

\begin{prop}Triangle and Star graphs are the only set of graphs which are k-connected with $k < 3$ and also non-squeezable.
\end{prop}

\begin{proof}
The proof is based on the following lemma.
\end{proof}

\noindent \texttt{LEMMA 3.1}. \textit{If $G = (V, E)$ is a non-squeezable graph, then $\forall$ separating pairs ($u, v$) in $G$,  $\forall$ $x$ in $V \setminus\{u, v\}$, $u$ separates $x$ and $v$. Or, $\forall$ separating pairs ($u, v$) in $G$, $\forall$ x in $V \setminus \{u, v\}$, v separates x and u.}

\begin{proof}
The proof is by contradiction. Let $G$ be a non-squeezable graph. Let us assume that for all the vertices except $x$, $u$ separates $x$ and $v$. Now $v$ separates $x$ and $u$. Or $x$ is directly connected to $v$. This means that we can squeeze the sub-graph $u-v-x$ to $u-x$. This contradicts our assumption that $G$ is a non-squeezable graph. Thus, the theorem is true. Proposition 2 follows from this theorem.
\end{proof}

\begin{prop}
 Any graph G can be squeezed into a single edge by successively squeezing all the non-squeezable sub-graphs in G.
\end{prop}

This proposition is trivial to prove. Thus non-squeezable graphs can be considered as the atomic blocks from which the graphs are constructed. 
\begin{defn}
 \textit{Squeeze Minor} is a non-squeezable sub-graph that we squeeze during the squeezing operation.
\end{defn}

\begin{prop}
 The multi-set of squeeze minors that are obtained by successive squeezing of non-squeezable graphs in a graph is unique.
\end{prop}

\begin{proof}
We give a sketch of the proof. The multi-set of squeeze minors can contain only stars, triangles and triconnected components. Note that the stars in this multi set will be unique since it corresponds to the cut vertices of the graph which are unique. Now, we need to prove that the triangles and the triconnected components in the multi set are unique. If we disconnect the graph at the cut vertices, we will get biconnected components. Consider an arbitrary sequence of squeezes which results in the squeezing of a biconnected component into one of the edges in a star. Consider some squeeze in that sequence. That squeeze is possible only since the component being squeezed is triconnected. The component was either tri-connected to begin with or became triconnected by the previous squeeze operations which introduced virtual edges. Recursively, these triconnected components are also unique. Same argument holds for triangles also.
\end{proof}

If $G_\mathcal{G}$ is the set of all graphs that can be generated using the RPERG $\mathcal{G}$, then from proposition 4, we can say that, for any $g \in G_\mathcal{G}$, $g$ cannot be generated by applying different sets of rules. 

\subsection{Learning the Grammar}

Now we will define a restricted version of PERG, called RPERG.
\begin{defn}
 A Restricted Probabilistic Edge Replacement grammar (RPERG) is a PERG such that for every rule $A\rightarrow R \in $ RPERG, R is a non-squeezable graph.
\end{defn}

RPERGs can be viewed as analogous to PCFGs in the string grammar literature, while PERGs are analogous to Tree Substitution Grammars (TSG). This is more intuitive in the sense that in PCFG, RHS of the rules can contain only a 2-level tree, while TSGs can contain any sub-tree as RHS. Similarly, RPERGs can contain only non-squeezable graph fragments in RHS, while PERGs can contain any graph fragment in RHS.

 In this section, we will see an algorithm to learn RPERG from a set of graphs. Let $D = (g_1, g_2, ...g_n)$ be a set of graphs. We assume that the graphs are undirected and the edges are of the same type and are un-weighted. Given this data, we need to learn the RPERG that could have generated this data. We also assume that all the edges in the data are non-terminal edges. So, all the rules will have only non-terminal edges.

Consider a graph $G$. Let $c(A\rightarrow R)$ be the count of the occurrence of the non-squeezable sub-graph $R$ in $G$. Now, the probability of this graph $G$ under an RPERG is given as,
\[p(G) = \prod_{A\rightarrow R \in P} p(A \rightarrow R)^{c(A \rightarrow R)}\]
For a model built on a set of graphs D, the maximum likelihood estimation of the parameters of the model is given by,
\[p_{ML}^{A \rightarrow R} = \frac{c_D(A \rightarrow R)}{\sum_{R':A\rightarrow R'} c_D(A \rightarrow R') }\]\\
where $c_D(A\rightarrow R)$ is the count of the occurrences of the sub-graph $R$ in the data $D$.

Now the learning problem has been reduced to getting the counts of non-squeezable components in a graph. Let us first consider a simple approach to count the non-squeezable components from a graph. We can first find a non-squeezable sub-graph, squeeze it and repeat the same until we squeeze the entire graph into an edge. But, finding a non-squeezable subgraph by repeated squeezing is computationally expensive. We will propose a more efficient algorithm to count all non-squeezable sub-graphs based on Proposition 1 and the intuitions given in Propositions 3 and 4.

The learning algorithm is given in Algorithm \ref{alg1} . In the algorithm, star(n) denotes a star network with n+1 nodes. The statement $C(A\rightarrow R)+ = 1$ increments the count of the rule $A\rightarrow R$, if it is already present in the set of learnt rules or it will add the rule to the rule set and set count to 1. We assume that the Stack data structure and the grammar rule set are shared between \textsc{Main} and \textsc{Get\_Components} functions.


Finding split pairs in a biconnected component is the most non-trivial step in the learning algorithm. Any naive implementation of this module would take $O(n^3)$ time. A linear time algorithm, which is linear in the size of the graph is provided by \cite{14}. We have used a publicly available implementation of this algorithm (\url{https://github.com/adrianN/Triconnectivity}) to find split pairs.

Note that the algorithm is inherently parallelizable. Once we split the graph, we can parallelly learn rules from individual sub-graphs. This will make the algorithm even faster. 

\begin{algorithm}[h!]
    \small
    \caption{Learn RPERG}\label{alg1}
    
     \textbf{Input:} Set of Graphs D = $\{g_1,g_2,...,g_n\}$. \newline
     \textbf{Output:} RPERG
    \begin{algorithmic}[1]
    \Function{Main}{Set of Graphs D}
    \State $\text{Stack} \gets $ empty stack
    \For{\textit{each graph $g_i$}}
        \State \textsc{Get\_Components}($g_i$)
        \While{\textit{Stack is not empty}}
           \State $\textit{g} \gets \text{Stack.pop()}$
           \State find a split pair (a,b) in \textit{g}
            \If{ $\exists$ \textit{no split pair}}
               \State $C(A\rightarrow g) += 1$
           \Else
             \State $g_1,g_2\gets $Obtained by splitting \textit{g} at (a,b)
               \If{\textit{edge(a,b)} $\notin g$}
                   \State Add edge(a,b) to $g_2$ 
               \EndIf
               \For{\text{\textit{g}' in $g_1,g_2$}}
                   \State \textsc{Get\_Components}(\textit{g'})
               \EndFor
            \EndIf
        \EndWhile
    \EndFor
   \EndFunction
   \end{algorithmic}
   \hfill
   \hfill
   \begin{algorithmic}[1]
   \Function{Get\_Components}{Graph \textit{g}}
   \State $\text{CV} \gets $ cut vertices in \textit{g}
        \For{\textit{each $v_k$ in CV}}
            \State \textit{n} $\gets$ no. of biconn. components connected by $v_k$
            \State \textit{C}(\textit{A $\rightarrow$ star}(\textit{n})) += 1
        \EndFor
        \State \textit{S} $\gets$ set of all biconnected components in \textit{g}
        \For{\textit{each $s_i$ in S}}
            \State \textit{Stack.push}($s_i$)
        \EndFor
   \EndFunction

    \end{algorithmic}
   
    \end{algorithm}



\subsection{The Generative Model}


In the previous section, we have seen an algorithm for learning RPERGs. Given a set of graphs, the learning algorithm will learn an RPERG from the graphs. In this section, we propose two different generative models for graphs based on the learnt RPERG.

Since we have assumed that the graphs have only one type of link, the rules contain only one non-terminal label, namely $A$. We consider the absence of a label for an edge as a terminal; in other words, $\epsilon$ is the only terminal. The learning algorithm will consider all the edges in the given graph to be non-terminal edges. So, in the learnt grammar, all the edges in the RHS of the rules will be non-terminal edges.

\subsubsection{Proposed Model - \textit{ERGM-1}}

The first model is based on grammar derivation. If we start with an edge labelled with A and apply the rules from the learnt grammar successively, the derivation will not terminate since none of the learnt rules have terminal edges in the RHS. So we append the learnt model with an additional rule which converts a labelled edge to an un-labelled edge. We assign probability $p$ to this rule and re-normalize the probability of other rules accordingly. We call this grammar, modified RPERG. Now, the generation of a new graph is nothing but the successive application of the rules until we get all terminal edges. The model is described in Algorithm \ref{alg3}. 

This model gives the notion of likelihood of a graph being a member of a particular class of graphs. Given a set of graphs belonging to a class, we can learn a grammar for the class and parse the given graph with that grammar. The probability of the graph gives us some idea about the membership of the graph to this class. Although we cannot exactly control the size of the graph with this model, the coarse size of the graph can be approximately controlled by the parameter $p$. 

\begin{algorithm}[h!]
\small
\caption{\textit{ERGM-1}}\label{alg3}

\textbf{Input:} Modified RPERG \newline
\textbf{Output:} A Graph

\begin{algorithmic}[1]
    \State Graph G = NULL
    \State Add a non-terminal edge to G
    \While{\textit{$\exists$ a non-terminal edge in G}}
        \State Randomly pick a non-terminal edge A in G.
        \State Sample a rule $A\rightarrow R$ from RPERG and replace A with R in G.
    \EndWhile
    \State return G
\end{algorithmic}

\end{algorithm}

\subsubsection{Proposed Model - \textit{ERGM-2}}

 The second model uses the learnt RPERG directly without doing any additions. We start with an edge and randomly choose an edge and replace it with a sub-graph based on rule sampled from the distribution of rules, until we get a graph of required size. By size, here we mean the number of nodes. Then we stop expanding and convert all non-terminal edges to terminal edges. The model is described in Algorithm \ref{alg4}. This model can be used to generate a graph of required size, with desired properties that are learnt from a set of graphs.
\begin{algorithm}[h!]
\small
\caption{\textit{ERGM-2}}\label{alg4}

\textbf{Input:} RPERG \newline
\textbf{Output:} A Graph

\begin{algorithmic}[1]
	
    \State Graph G = NULL
    \State Add a non-terminal edge to G
    \While{\textit{desired graph size is not reached}}
        \State Randomly pick a non-terminal edge A in G.
        \State Sample a rule $A\rightarrow R$ from RPERG and replace A with R in G.
    \EndWhile
    \State Convert all non-terminal edges to terminal edges.
    \State return G
\end{algorithmic}

\end{algorithm}

\section{Experiments}
Here, we show that RPERGs contain rules that capture the structure of the graph. We test our proposed model by fitting it onto several real life graphs. First, we learn the grammar from the graph. Then, we generate graphs from the learnt grammar using the generative model. In our experiments, we use $ERGM$-$2$ to generate the graphs. In this section, we compare our approach against existing state-of-the-art graph generators.

\subsection{Real World Datasets}
The datasets considered in this paper are the same as those used in \cite{aguinaga2016growing}. The networks vary not only in the number of vertices and edges, but also in the clustering coefficient, diameter, degree distribution and many other graph properties. Table \ref{tab:tab1} gives the statistics for these networks. The Arxiv GR-QC covers scientific collaborations in the General Relativity and Quantum Cosmology section of Arxiv; the Internet Routers is network of autonomous systems of the internet connected with each other; Enron Emails is email correspondence graph of Enron corporation; DBLP is co-authorship graph from DBLP dataset. The graphs were obtained from \href{http://snap.stanford.edu/data}{SNAP}  and \href{http://konect.uni-koblenz.de}{KONECT} repositories.

\begin{table}[h!]
    \small
    \centering
    
    \begin{tabular}{|c|c|c|c|c|}
    \hline
        Dataset & Nodes & Edges & Diameter & Clust. Coeff.\\
        \hline
        Arxiv & 5242 & 14496 &17&0.529 \\
        \hline 
        Routers&6474 & 13895 &9 & 0.252\\
        \hline
        Enron & 36692 & 183831 &11 &0.497 \\
        \hline
        DBLP & 317080 & 1049866 &21 &0.632 \\
        \hline
    \end{tabular}
    \caption{\small Dataset Statistics for real world graphs.}
    \label{tab:tab1}
\end{table}

\subsection{Comparison with existing models}
We compare several properties of graphs from four different graph generators (RPERG, HRG\cite{aguinaga2016growing}, Chung-Lu\cite{chung2002connected}, Kronecker\cite{leskovec2010kronecker}) with the original graph G. The HRG based approach has already been introduced in Section \ref{sec:hrg}. The Chung-Lu model takes a degree distribution as input and generates a new graph with similar degree distribution and size. Kronecker model first learns an initiator matrix and then performs a recursive multiplication of that initiator matrix to create an adjacency matrix of the approximate graph. We use KronFit\cite{leskovec2010kronecker} to learn the $2\times2$ initiator matrix. 

In this section, we generate 10 graphs each for the graph generators and plot the mean values for different properties. Figures \ref{fig:fig4},\ref{fig:fig5},\ref{fig:fig6} contain plots of graph properties (Degree Distribution, Network Values, Hop Plot, Mean Clustering Coefficient, Scree Plot, Node Triangle Participation) for Arxiv, Routers and Enron datasets respectively. 
\begin{figure}[h!]
	\centering
	\captionsetup{skip=-6pt}
	\begin{subfigure}[b]{.235\textwidth}
		\centering
		\includegraphics[scale=0.285]{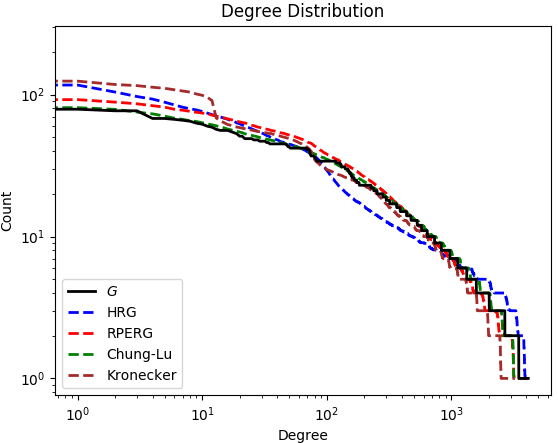}
	    \label{fig1a}
	\end{subfigure}
	\begin{subfigure}[b]{.215\textwidth}
		\centering
		\includegraphics[scale=0.285]{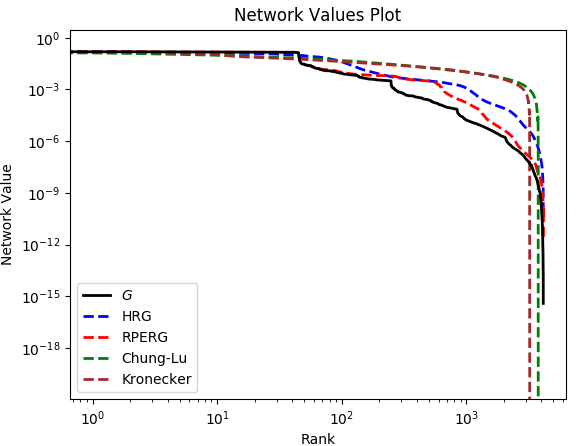}
	    \label{fig1b}
	\end{subfigure}
	\begin{subfigure}[b]{.235\textwidth}
		\centering
		\includegraphics[scale=0.285]{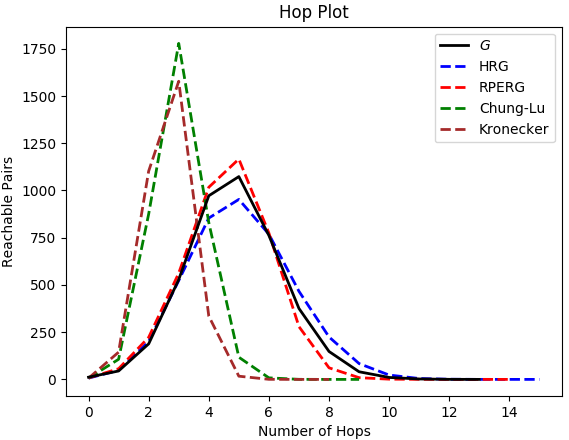}
	    \label{fig1c}
	\end{subfigure}
	\begin{subfigure}[b]{.215\textwidth}
		\centering
		\includegraphics[scale=0.285]{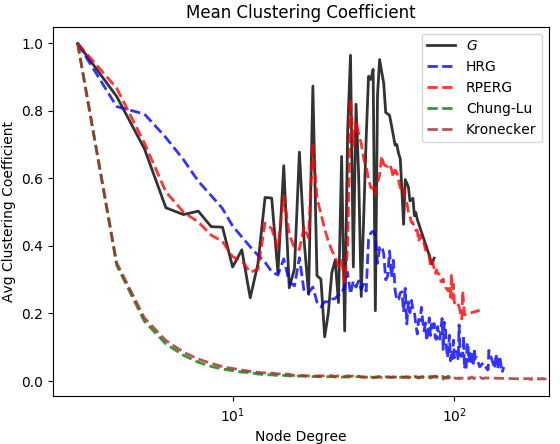}
	    \label{fig1d}
	\end{subfigure}
	\begin{subfigure}[b]{.235\textwidth}
		\centering
		\includegraphics[scale=0.285]{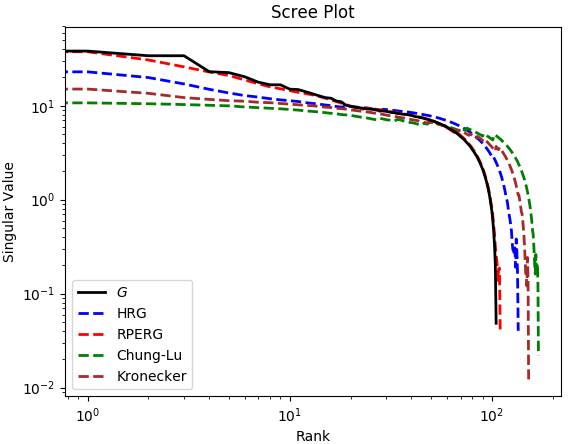}
	    \label{fig1d}
	\end{subfigure}
	\begin{subfigure}[b]{.215\textwidth}
		\centering
		\includegraphics[scale=0.285]{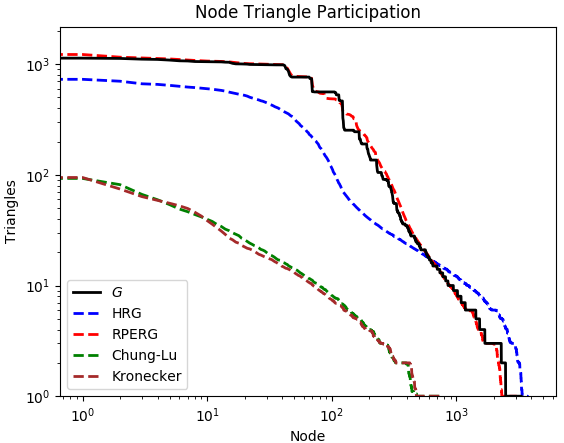}
	    \label{fig1d}
	\end{subfigure}

	\caption{\small Plots for Arxiv GR-QC dataset.}
    \label{fig:fig4}
\end{figure}

\noindent\textbf{Degree Distribution:} It is the distribution of number of edges connecting to each vertex in the graph. From the plots, we can see that each of the generators give graphs that are slightly different from original graph, but all of the them capture the power law degree distribution. \\
\textbf{Network Values:} This is a plot of the eigen components of the eigen vector corresponding to the largest eigen value as a function of their rank. From the plots, it can be seen that RPERG performs consistently well across all graphs but the difference between generators is difficult to discern. To more concretely compare the eigenvectors, the cosine distance between the eigenvector centrality of the original graph and the model's generated graphs is shown in Table \ref{tab:tab2}. It can be seen that the distance values are lowest for RPERG.
\begin{table}[ht!]
    \small 
    \centering
    \setlength{\belowcaptionskip}{-4pt}
    \begin{tabular}{|c|c|c|c|c|}
    \hline
    Dataset & RPERG & HRG & Chung-Lu & Kronecker\\
    \hline
    Arxiv  & \textbf{0.0025} & 0.0161 & 0.3496& 0.3406\\
    \hline
     Routers &\textbf{0.0247} & 0.0411 & 0.0379  & 0.0614\\
    \hline
    Enron  & \textbf{0.00007} & 0.0002 & 0.0052 &0.0676\\
    \hline
    DBLP & \textbf{0.0079} & 0.0649 & 0.5854 & 0.4997\\
    \hline 
    
    \end{tabular}
    \caption{\small Cosine Distance between the eigenvector centrality of original graph and graphs from generator.}
    
    \label{tab:tab2}
\end{table}

\noindent\textbf{Hop Plot: }Hop plot shows the number of vertex-pairs that are reachable within $x$ hops. It gives a sense about the distribution of the shortest path lengths in the network and about how quickly nodes' neighborhoods expand with the number of hops. Similar to \cite{leskovec2007graph,aguinaga2016growing}, we generate hop plot by choosing 50 random nodes and performing a complete breadth-first traversal over each graph. From the plots, we can see that hop-plots of RPERG are consistently similar to the original graph.\\
\textbf{Mean Clustering Coefficient: }Clustering coefficient is one particular measure of community structure that has been widely used in literature\cite{prat2014community,kolda2014scalable}. We plot the average clustering coefficient of the nodes as a function of its degree in the graph. From the plots, it can be seen that RPERG matches the community structure of the original graph. Similar to \cite{seshadhri2012community}, we see that Chung-Lu and Kronecker models perform poorly in this task.\\
\textbf{Scree plot: }This is a plot of the eigen values of the graph adjacency matrix as a function of their rank, which has been found to obey power law\cite{chakrabarti2004r}. From the plots, we can see that RPERG is closest to original graph in terms of eigen distributions.

\noindent\textbf{Node Triangle Participation: }This is a plot of the number of triangles versus the number of nodes that participate in that triangles. It is a measure of transitivity in networks\cite{leskovec2010kronecker} since edges in real-world networks tend to cluster\cite{9} and form triads of connected nodes. From the plots, it can be seen that RPERG consistently captures the node triangle participation of the original graph. \\
\textbf{Graphlet Correlation Distance (GCD): }\cite{yaverouglu2015proper} has identified a new metric called GCD. It computes the distance between two graphlet correlation matrices. GCD measures the frequency of the various graphlets present in each graph, i.e the number of edges, wedges, triangles, squares, 4-cliques, etc., and compares the graphlet frequencies between two graphs. Because GCD is a distance metric, lower values are better. Table \ref{tab:tab3} compares the GCD of original graph with the graphs generated using RPERG, HRG, Chung-Lu and Kronecker. It can be seen that GCD values are lowest for our model.
\begin{table}[h!]
    \small
    \centering
    
    \begin{tabular}{|c|c|c|c|c|}
    \hline
    Dataset & RPERG & HRG & Chung-Lu & Kronecker\\
    \hline
    Arxiv  & \textbf{1.086} &1.094 &1.792 &2.071 \\
    \hline
   Routers &\textbf{1.293} &1.404 &1.975 &2.776 \\
    \hline
    Enron  & \textbf{0.487} & 0.525 &1.319 & 2.83\\
    \hline
    DBLP & \textbf{0.409} & 1.602 & 1.738 & 2.821 \\
    \hline 
    
    \end{tabular}
    \caption{\small Graphlet Correlation Distance values.}
    \label{tab:tab3}
\end{table}

\begin{figure}[ht!]
	\centering
	\captionsetup{aboveskip=-6pt, belowskip=11pt}
	\begin{subfigure}[b]{.235\textwidth}
		\centering
		\includegraphics[scale=0.275]{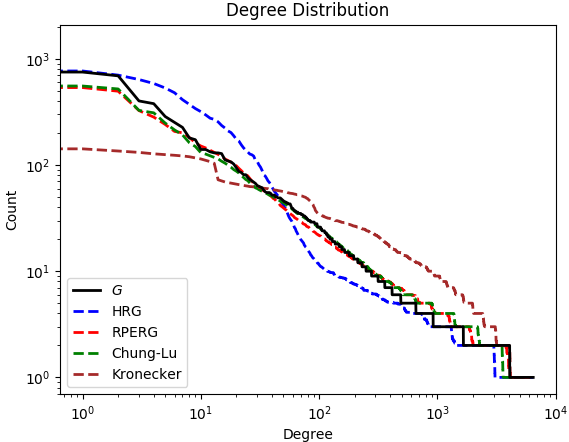}
	    \label{fig1a}
	\end{subfigure}
	\begin{subfigure}[b]{.215\textwidth}
		\centering
		\includegraphics[scale=0.275]{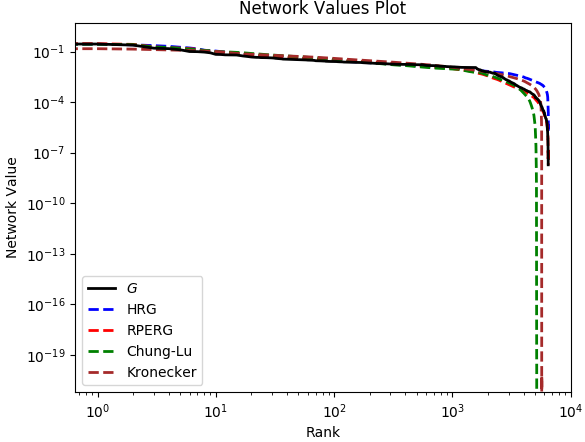}
	    \label{fig1b}
	\end{subfigure}
	\begin{subfigure}[b]{.235\textwidth}
		\centering
		\includegraphics[scale=0.275]{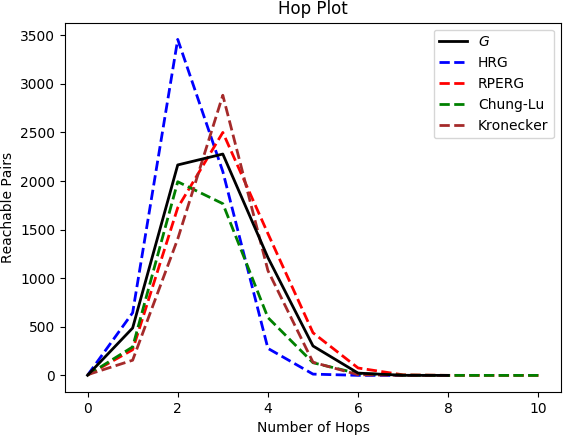}
	    \label{fig1c}
	\end{subfigure}
	\begin{subfigure}[b]{.215\textwidth}
		\centering
		\includegraphics[scale=0.275]{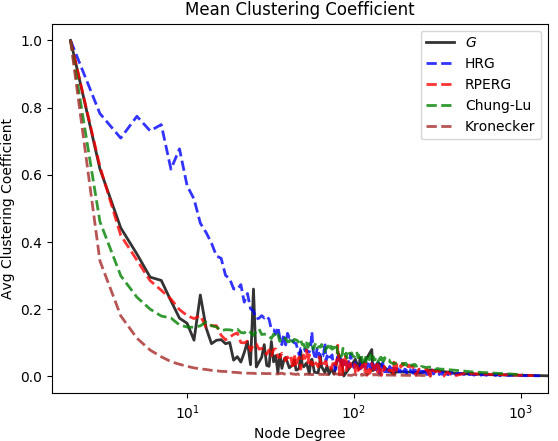}
	    \label{fig1d}
	\end{subfigure}
	\begin{subfigure}[b]{.235\textwidth}
		\centering
		\includegraphics[scale=0.275]{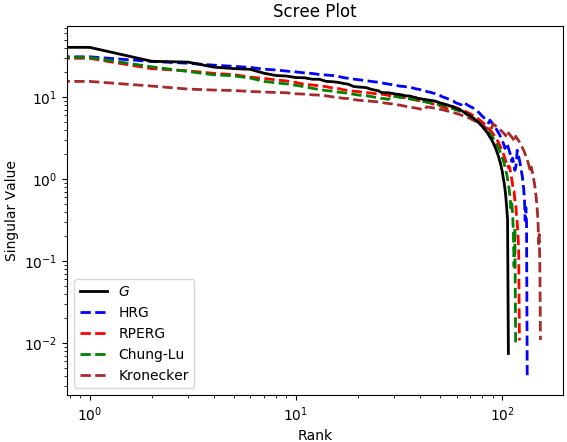}
	    \label{fig1d}
	\end{subfigure}
	\begin{subfigure}[b]{.215\textwidth}
		\centering
		\includegraphics[scale=0.275]{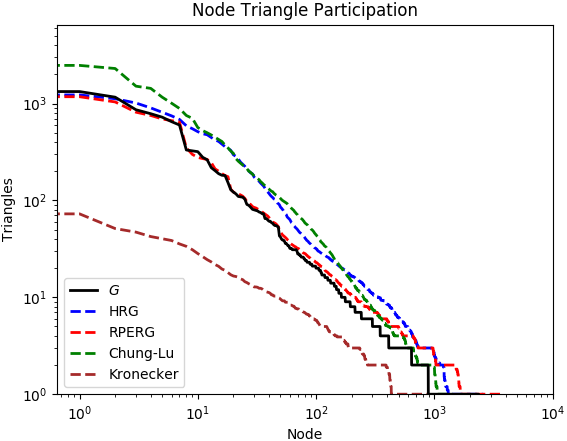}
	    \label{fig1d}
	\end{subfigure}

	\caption{\small Plots for Internet Routers dataset.}
    \label{fig:fig5}
\end{figure}

\begin{figure}[ht!]
	\centering
	\captionsetup{aboveskip=-6pt, belowskip=-25pt}
	\begin{subfigure}[b]{.235\textwidth}
		\centering
		\includegraphics[scale=0.280]{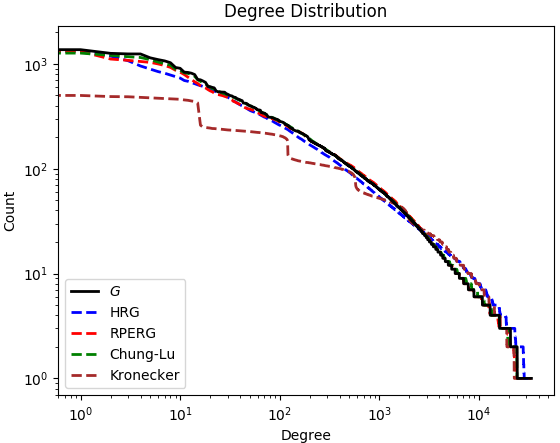}
	    \label{fig1a}
	\end{subfigure}
	\begin{subfigure}[b]{.215\textwidth}
		\centering
		\includegraphics[scale=0.280]{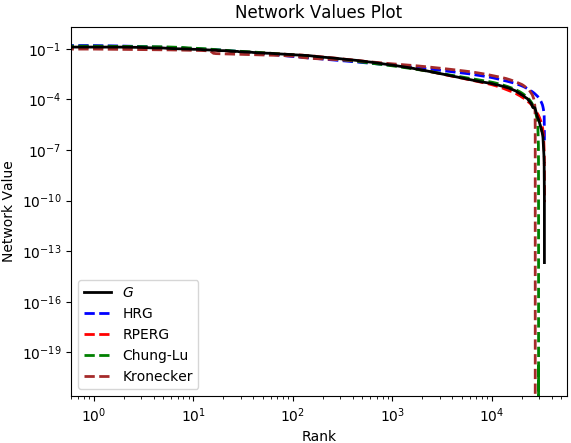}
	    \label{fig1b}
	\end{subfigure}
	\begin{subfigure}[b]{.235\textwidth}
		\centering
		\includegraphics[scale=0.280]{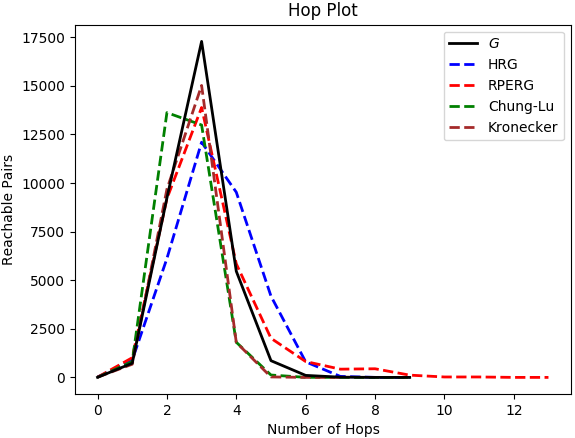}
	    \label{fig1c}
	\end{subfigure}
	\begin{subfigure}[b]{.215\textwidth}
		\centering
		\includegraphics[scale=0.280]{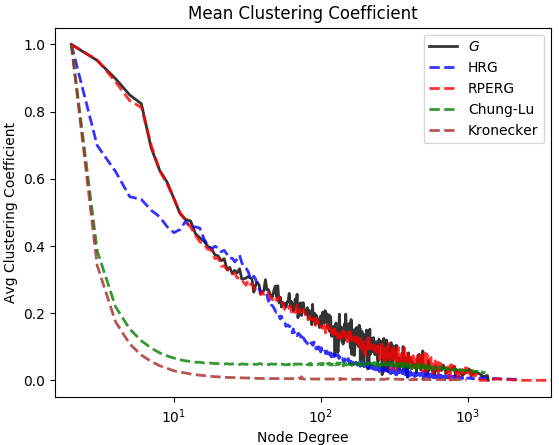}
	    \label{fig1d}
	\end{subfigure}
	\begin{subfigure}[b]{.235\textwidth}
		\centering
		\includegraphics[scale=0.280]{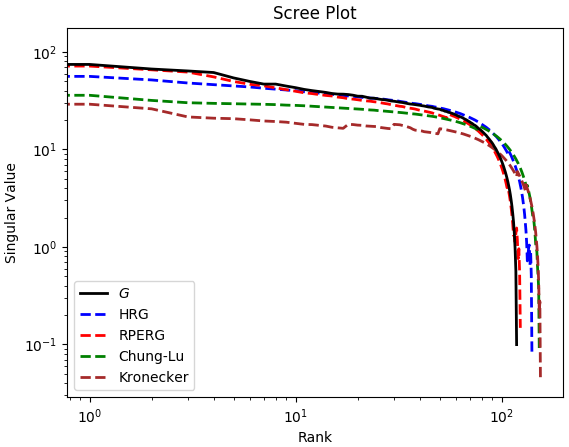}
	    \label{fig1d}
	\end{subfigure}
	\begin{subfigure}[b]{.215\textwidth}
		\centering
		\includegraphics[scale=0.280]{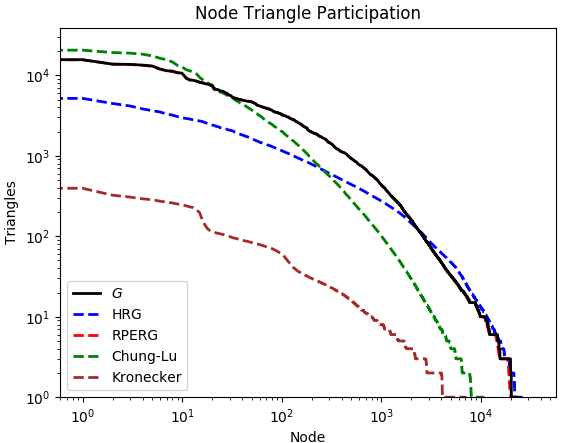}
	    \label{fig1d}
	\end{subfigure}

	\caption{\small Plots for Enron Emails dataset.}
    \label{fig:fig6}
\end{figure}

\subsection{Runtime Analysis}
The overall runtime of the RPERG model can be split into two parts: (1) Rule extraction, and (2) Graph generation.
Let the given graph G contain n vertices and m edges. Each iteration of the \textit{while} loop in line 5 of Alg.\ref{alg1} requires $O(n+m)$ time for finding a split pair and $O(n+m)$ time for finding all cut-vertices.
The runtime of RPERG learning process depends on the type of split obtained at each iteration. In the worst case, size of the graph reduces by 1 node, after splitting it, at each iteration, and time complexity is $O(m \cdot n)$. Conversely, best case time complexity is $O(m \cdot log n)$.
By comparison, HRG rule extraction takes $O(m \cdot \Delta)$ time, where $\Delta$ is maximum degree of G, Kronecker learns model in $O(m)$, Chung-Lu does not learn a model, but takes the degree sequence as input.

For RPERG and HRG, since graph generation is a straightforward application of the grammar rules, the time complexity is linear in the number of edges of the output graph. For Kronecker, graph generation is in $O(m)$ whereas it takes $O(n+m)$ for Chung-Lu model.

\section{Conclusion}
We propose a graph generation model based on probabilistic graph grammars. We characterize the notion of non-squeezable graphs and restrict our attention only to edge replacement rules that introduce non-squeezable components. From our experiments, we find that the graphs generated by our model more closely resemble the original graph compared to those obtained by existing graph generators. Even though the grammar is context free, it is able to capture most of the statistical properties of the graph. We observe that our algorithm is easily parallelizable as we can simultaneously run on all the graphs when we have multiple graphs. 

There are several extensions for the model that are possible. We can try to model preferential attachment by using a context sensitive grammar. Tackling graphs with multiple types of links (heterogeneous links) is also a challenging problem. In our algorithm, we stop finding split pairs when we find the first split pair. This can be improved further by not stopping and continue finding a pair which splits the graphs into reasonable two halves.

\bibliographystyle{siam}
\def\bibliographytypesize{\tiny}
\bibliography{references}

\end{document}